\documentclass[11pt]{article}

\usepackage[all]{xy}           
\usepackage{amsmath,amsfonts,amssymb,amsthm,textcomp}
\usepackage{eucal}
\usepackage{epsfig}

\def\v #1{\vert #1\vert}             

\def\m #1 #2{(-1)^{{\v #1} {\v #2}}} 

\theoremstyle{plain}
\newtheorem{theorem}{Theorem}

\newtheorem{proposition}[theorem]{Proposition}

\theoremstyle{definition}
\newtheorem{definition}[theorem]{Definition}

\def\<#1>{\langle#1\rangle}

\usepackage{epsfig}
\usepackage{amsmath}
\usepackage{epic}
\usepackage{amssymb}
\usepackage{fancybox}
\usepackage{wrapfig}
\usepackage[font=footnotesize,labelsep=period]{caption}

\usepackage{float}

\usepackage{booktabs}

\usepackage[all]{xy}

\usepackage{fancyhdr}

\usepackage{fancyhdr}

\usepackage{appendix}

\begin{document}

\centerline{\Large \bf A geometric approach to solve}\vskip 0.25cm
\centerline{\Large \bf  time dependent and dissipative Hamiltonian systems}

\medskip
\medskip

\centerline{M. de Le\'on and C.
Sard\'on}
\vskip 0.5cm
\centerline{Instituto de Ciencias Matem\'aticas, Campus Cantoblanco}\vskip 0.2cm
\centerline{Consejo Superior de Investigaciones Cient\'ificas}
\vskip 0.2cm
\centerline{C/ Nicol\'as Cabrera, 13--15, 28049, Madrid. SPAIN}

\begin{abstract}

In this paper, we apply the geometric Hamilton--Jacobi theory to obtain solutions of Hamiltonian systems in Classical Mechanics,
that are either compatible with a cosymplectic or a contact structure. As it is well known, the first structure plays a central role
in the theory of time-dependent Hamiltonians, whilst the second is here used to treat classical Hamiltonians including dissipation terms.

On the other hand, the interest of a geometric Hamilton--Jacobi equation is the primordial observation that a Hamiltonian vector field $X_{H}$ can be projected into the configuration manifold by means of a 1-form $dW$, then the integral curves of the projected
vector field $X_{H}^{dW}$can be transformed into integral curves of $X_{H}$ provided that $W$ is a solution of the Hamilton--Jacobi equation.

In this way, we use the geometric Hamilton--Jacobi theory to derive solutions of physical systems with a Hamiltonian formulation. A new expression for a geometric
Hamilton Jacobi equation is obtained for time dependent Hamiltonians described with the aid of a cosymplectic structure. Then, another
expression for the Hamilton Jacobi equation is retrieved for Hamiltonians with frictional terms described through contact geometry.
Both approaches shall be applied to physical examples. 
\end{abstract}

\section{Introduction}
%

In this paper we are concerned with {almost cosymplectic structures} and their application in classical Hamiltonian Mechanics.
By an {\it almost cosymplectic structure} we understand a $2n+1$-dimensional manifold equipped with a one-form $\eta$ and a two-form $\Omega$ such that $\eta\wedge \Omega^n$ is a volume form.
In particular, we will study the case of {\it cosymplectic manifolds} \cite{Blair,Cape,LeonSara,LeonTuyn} and {\it contact manifolds} \cite{Blair,boothwang,etnyre,Godbillon}.
Cosymplectic manifolds have shown their usefulness in theoretical Physics, as in gauge theories of gravity, branes and string theory \cite{Becker,Cham,Hitchin}. 
Among the early studies of cosymplectic manifolds we mention A. Lichnerowicz \cite{Lich2,Lich}, who studied the Lie algebra of infinitesimal automorphisms of a 
cosymplectic manifold, in analogy with the symplectic case. 
But since very foundational papers also by P. Libermann, very sporadic papers have appeared on cosymplectic settings. It is a current act of will to provide surveys on cosymplectic geometry due to their lack \cite{Cape,LeonChineaMarrero}. Posteriously,
some works have endowed cosymplectic manifolds with a Riemannian metric, what made of them the so-called {\it coK\"ahler manifolds} \cite{oku}. These
are the odd dimensional counterpart of K\"ahler manifolds.
Another important role of the cosymplectic theory is the reduction theory to reduce time dependent Hamiltonians by symmetry groups \cite{Albert,CLMMdD,CLM}.

Our particular interest in cosymplectic structures resides in their usage in the description of time dependent Mechanics. They are present
in numerous formulations of classical regular Lagrangians \cite{LeonMarinMarrero}, Hamiltonian systems \cite{Laci} or Tulczyjew-like descriptions \cite{LeonMarrero} in terms
of Lagrangian submanifolds \cite{LMM}.

Nonetheless, there have been more written monographs on contact geometry. The interest in contact structures roots in their applications 
in partial differential equations as in Thermodynamics \cite{rajeev} or Geometric Mechanics \cite{IboLeonMarmo},
geometric Optics \cite{CariNasarre,hamil1,hamil2}, geometric quantization \cite{rajeev} and applications to low dimensional topology, as it can be the characterization
of Stein manifolds \cite{Otto,Stein}.
Also, the theory of contact structures is linked to many other geometric backgrounds, as it is the case of symplectic geometry, Riemannian and complex geometry,
analysis and dynamics \cite{Blair, Godbillon}.

For both approaches, we give central importance to the construction of a vector field (said to be Hamiltonian) with a corresponding smooth function with respect
to a contact structure or a cosymplectic structure, for its subsequent use in the geometric Hamilton--Jacobi theory (HJ theory).

The HJ theory has proven its popularity given its equivalence to other theories of Classical Mechanics,
and its simple principal idea: a Hamiltonian vector field $X_{H}$ can be projected into the configuration manifold by means of a 1-form $dW$, then the integral curves of the projected
vector field $X_{H}^{dW}$ can be transformed into integral curves of $X_{H}$ provided that $W$ is a solution of the Hamilton--Jacobi equation \cite{Arnold,Gold,Kibble,LL,LiberMarle,Rund}. In the last decades, the HJ theory has been interpreted in modern geometric terms \cite{CGMMMLRR,CGMMMLRR1,LeonIglDiego,MarrSosa,Pauffer}
and has been applied in multiple settings: as nonholonomic \cite{CGMMMLRR,CGMMMLRR1,LeonIglDiego,LeonMarrDiego}, singular Lagrangian Mechanics \cite{LeonMarrDiegoVaq,LeonDiegoVaq2} and classical field theories \cite{LeonMarrDiego08, Cedric}.
This theory relies in the existence of lagrangian/legendrian submanifolds. The notion of lagrangian/legendrian submanifolds has gained a lot of applications in dynamics from their introduction by Tulczyjew \cite{tulzcy1,tulzcy2}. We show how these submanifolds
are used to extend the geometric theory of the Hamilton--Jacobi equation from different geometric backgrounds. We use particular cases of lagrangian/legendrian submanifolds
in different geometric frameworks.

%
%
%

The paper is organized as follows: Section 2, is dedicated to review fundamentals on Geometric Mechanics and notation which will be used throughout the paper.
In Section 3, we introduce the dynamics on contact and cosymplectic manifolds and illustrate their geometric characteristics.
Section 4 contains the theory of Lagrangian--Legendrian submanifolds which will be used along the forthcoming sections.
In Section 5, we propose a geometric Hamilton--Jacobi theory on cosymplectic manifolds and illustrate our result with an example, a case of time dependent
Hamiltonian system. In particular, this system is the well-known Winternitz--Smorodinsky oscillator, for which we obtain an explicit expression for the solution $\gamma$
of the Hamilton--Jacobi equation.
In similar fashion, we devote Section 6 to propose a geometric Hamilton--Jacobi equation for a contact manifold. We also illustrate our result through an example,
which is the case of a classical Hamiltonian with kinetic and potencial term, accompanied by a dissipative term.




\section{Geometric Mechanics: Fundamentals}
We hereafter assume all mathematical objects to be $C^{\infty}$ and globally defined. Manifolds are considered connected. This permit us to
omit technical details while highlighting the main aspects of our theory.


A classical Hamilton system is given by a Hamilton function $H(q^i,p_i)$, where $q^i$ are the positions in a configuration manifold $Q$ and $p_i$ are the conjugated
momenta. The Hamiltonian can be interpreted as total energy of the system $H=T+V$.
We compute the differential of the function, 
$$dH=\frac{\partial H}{\partial q^i}dq^i+\frac{\partial H}{\partial p_i}dp_i$$
 and write the equation
\begin{eqnarray*}
\left. X_H = \left( \begin{array}{ccc}
0&I_n\\
-I_n&0 \end{array} \right) \left(
\begin{array}{c}
\frac{\partial H}{\partial q^i}\\
\frac{\partial H}{\partial p_i}
\end{array}
\right)
\right.
\end{eqnarray*}
where $I_n$ is the identity matrix of order $n$. The above matrix is called a symplectic matrix.
The vector field $X_H$ is called a Hamiltonian vector field and its integral curves $(q^i(t),p_i(t))$ are the Hamilton equations. 
\begin{equation}\label{hamileq7}
\left\{\begin{aligned}
 {\dot q}^i&=\frac{\partial H}{\partial p_i},\\
 {\dot p}_i&=-\frac{\partial H}{\partial q^i}
 \end{aligned}\right.
 \end{equation}
for all $i=1,\dots,n.$
We can define a Poisson bracket of two functions as 
$$\{f,g\}=\sum_{i=1}^n\left(\frac{\partial}{\partial q^i}\frac{\partial}{\partial p_i}-\frac{\partial}{\partial p_i}\frac{\partial}{\partial q^i}\right)$$
which is bilinear, skew symmetric and fulfils the Jacobi identity
\begin{equation*}
 \{f,\{g,h\}\}+\{f,\{g,h\}\}+\{h,\{f,g\}\}=0,\quad \forall f,g,h\in C^{\infty}(Q)
\end{equation*}
\noindent
The symplectic two form 
\begin{equation}\label{symp}
\omega_Q=dq^i\wedge dp_i
\end{equation}
 has the associated symplectic matrix above. It is skew-symmetric and closed.
We can rewrite the Hamilton equations \eqref{hamileq7} in a compact, geometric form
\begin{equation}
 \iota_{X_H}\omega_Q=dH
\end{equation}

A {\it symplectic manifold} is a pair $(M,\omega)$ such that the two-form $\omega$ is regular (that is, $\omega^n\neq 0$) and closed. Then, $M$ has even dimension, say $2n$. 
The Darboux theorem states that given a symplectic manifold $(M,\omega)$ we can find Darboux coordinates $(q^i,p_i)$ such that the symplectic form
is written as in \eqref{symp}.
Indeed, any symplectic manifold is locally equivalent to the cotangent bundle $T^{*}Q$ of a configuration manifold $Q$.

Given a configuration manifold $Q$, its cotangent bundle $T^{*}Q$ is the phase space.
We consider the canonical projection $\pi_Q:T^{*}Q\rightarrow Q$. 
From the Poisson bracket, we can define a canonical 2-contravariant tensor such that $\Lambda_Q(df,dg)=\{f,g\}$, for all $f,g\in C^{\infty}(T^{*}Q)$. In Darboux coordinates it reads
\begin{equation}
\Lambda_Q=\sum_{i=1}^n \frac{\partial}{\partial q^i}\wedge \frac{\partial}{\partial p_i}
\end{equation}
that we call a Poisson bivector. It is the contravariant version of the symplectic form. Furthermore, on $T^{*}Q$, we consider the so-called Liouville form $\theta_Q=p_idq^i$ such that $\omega_Q=-d\theta_Q$.

\medskip

Now, we briefly recall some main aspects of Lagrangian Mechanics.
Let $L(q^i,\dot{q}^i)$ be a Lagrangian function, with $(q^i)$ the generalized coordinates on the manifold $Q$ and $(\dot{q}^i)$ are the generalized momenta. The Hamilton's principle
produces the Euler--Lagrange equations
\begin{equation}
 \frac{d}{dt}\left(\frac{\partial L}{\partial \dot{q}^i}\right)-\frac{\partial L}{\partial q^i}=0, \quad \forall i=1,\dots,n.
\end{equation}
A geometric version of these equations can be obtained.  We consider the tangent bundle $TQ$ and the canonical projection $\tau_Q:TQ\rightarrow Q$.
Consider a Lagrangian function $L:TQ\rightarrow \mathbb{R}$ and we define the vertical endomorphism $S=\frac{\partial}{\partial \dot{q}^i}\otimes dq^i$.
We have the Cartan 1-form $\theta_L=S^{*}(dL)=\frac{\partial L}{\partial {\dot q}^i}dq^i$
and the Cartan two-form $\omega_L=-d\theta_L$, and the energy
$$E_L=\Delta(L)-L\in C^{\infty}(TQ).$$
The operator $\Delta$ is defined as $\Delta=\dot{q}^i\frac{\partial}{\partial \dot{q}^i}$, and it is known as the {\it Liouville vector field} \cite{Crampin,Klein}.
From here, we recover the classical expressions
$$\omega_L=dq^i\wedge dp_i,\quad \text{such that}\quad p_i=\frac{\partial L}{\partial \dot{q}^i},\quad E_L=\dot{q}^ip_i-L.$$

We define the Hessian matrix
\begin{equation}
 \left(W_{ij}\right)=\left(\frac{\partial^2 L}{\partial \dot{q}^i \partial \dot{q}^j}\right)
\end{equation}
The Lagrangian $L$ is said to be regular if the Hessian matrix is invertible. 
In this case, the Lagrange equations can be written as
\begin{equation}
 \iota_{\xi_L}\omega_L=dE_L
\end{equation}
whose solution $\xi_L$ is called a Euler--Lagrange vector field. 
The vector field $\xi_L$ is a second-order differential equation which implies that its integral curves are tangent lifts of their projections
on the configuration manifold $Q$. These projections are called the solutions of $\xi_L$ and are just the solutions of the Euler--Lagrange equations \cite{LeonRodri}.

We denote the Legendre transformation as $FL:TQ\rightarrow T^{*}Q$ the fibered mapping, that is $\pi_Q\circ FL=\tau_Q$. We say that the Lagrangian is hyperregular if the Legendre transform $FL(q^i,\dot{q}^i)=(q^i,p_i)$,
where $p_i=\frac{\partial L}{\partial \dot{q}^i}$ is the conjugate momenta, is a global diffeomorphism.
This is the usual case in Mechanics, where $L=T-V$ with $T$ being the kinetic energy defined by a Riemannian metric $g$ on $Q$ and $V:Q\rightarrow \mathbb{R}$ is the potential.
Then, the Hamiltonian is simply retrieved as $H=E_L \circ FL^{-1}$.

The Hamiltonian vector field is obtained by $X_H=FL_{*}(\xi_L)$ and fulfills $\iota_{X_H}\omega_Q=dH$.
In this way, we establish the connection between the Euler--Lagrange equations and the Hamilton equations.

\medskip

\medskip
The Hamilton equations \eqref{hamileq7} can be equivalently be solved with the aid of the Hamilton--Jacobi equation.
 The Hamilton Jacobi theory consists on finding a {\it principal function}  $S(t,q^i)$, that fulfils
\begin{equation}\label{tdepHJ}
 \frac{\partial S}{\partial t}+H\left(q^i,\frac{\partial S}{\partial q^i}\right)=0
\end{equation}
where $H=H(q^i,p_i)$ is the Hamiltonian function of the system. Equation \eqref{tdepHJ} is referred to as the {\it Hamilton--Jacobi equation.}
 If we set the principal function to be separable in the
time variable $S=W(q^1,\dots,q^n)-Et$
where $E$ is the total energy of the system, then \eqref{tdepHJ} will now read \cite{AbraMars,Gold}

\begin{equation}\label{HJeq1}
 H\left({q}^i,\frac{\partial W}{\partial {q}^i}\right)=E.
\end{equation}
The Hamilton--Jacobi equation is a useful intrument to solve the Hamilton equations for $H$.
Indeed, if we find a solution $W$ of \eqref{HJeq1}, then any solution of the Hamilton
equations is retrieved by taking ${p}_i=\partial W/\partial {q}^i.$

Our aim in this paper is the geometric Hamilton--Jacobi approach.
Geometrically, the HJ theory can be reformulated as follows. If a Hamiltonian
vector field $X_{H}$ can be projected into the configuration manifold by means of a 1-form $dW$, then the integral curves of the projected
vector field $X_{H}^{dW}$can be transformed into integral curves of $X_{H}$ provided that $W$ is a solution of \eqref{HJeq1}. 
This explanation can be represented by the following diagram

\[
\xymatrix{ T^{*}Q
\ar[dd]^{\pi} \ar[rrr]^{X_H}&   & &TT^{*}Q\ar[dd]^{T\pi}\\
  &  & &\\
 Q\ar@/^2pc/[uu]^{dW}\ar[rrr]^{X_H^{dW}}&  & & TQ}
\]

\bigskip

This implies that $(dW)^{*}H=E$, with $dW$ being a section of the cotangent bundle. In other words, we are looking for a section $\alpha$ of $T^{*}Q$ 
such that $\alpha^{*}H=E$. As it is well-known, the image of a one-form is a lagrangian submanifold of $(T^{*}Q, \omega_Q)$ if and only if $d\alpha=0$ \cite{AbraMars}.
That is, $\alpha$ is locally exact, say $\alpha=dW$ on an open subset around each point.


\section{Cosymplectic and contact structures}
A {\it Jacobi structure} is the triple $(M,\Lambda, Z)$, where $Z$ is a vector field and $\Lambda$ is a
skew-symmetric bivector such that they fulfil the following integrability conditions
\begin{equation}\label{jaccond}
 [\Lambda,\Lambda]=2Z\wedge \Lambda,\qquad \mathcal{L}_Z\Lambda=0.
\end{equation}
We have the morphism $\sharp:T^{*}M\rightarrow TM$ defined as $\langle \sharp(\alpha),\beta \rangle=\Lambda(\alpha,\beta)$, for $\alpha,\beta \in \Omega^{1}$.
Vector fields associated with functions  $f\in \mathcal{F}(M)$ are defined as
\begin{equation*}
 X_f=\sharp(df)+fZ,
\end{equation*}
where we have denoted by $\mathcal{F}(M)$ the algebra of smooth functions on $M$.
The {\it characteristic distribution} $\mathcal{C}$ of $(M,\Lambda,Z)$ is a subset of $TM$ generated by the
values of all the vector fields $X_f$. This characteristic distribution $\mathcal{C}$ is defined by $\Lambda$ and $Z$, that is,
\begin{equation}
 \mathcal{C}_p=\sharp_p(T^{*}_pM)+<Z_p>,\quad \forall p\in M
\end{equation}
where $\sharp_p:T_p^{*}M\rightarrow T_pM$.
The fiber $\mathcal{C}_p=\mathcal{C}\cap T_{p}M$ of the characteristic distribution $\mathcal{C}$ over $p$ is the vector subspace of $T_pM$ generated by
$Z(p)$ and the image of the linear mapping $\sharp:T_p^{*}M\rightarrow T_pM$. The distribution is said to be transitive if the characteristic distribution
is the whole tangent bundle $TM$. Transitive Jacobi manifolds, according to the parity of their dimension, are either locally conformally symplectic or equipped
with a contact one form. They also include Poisson manifolds as particular cases \cite{Marle}. The Poisson bracket is a derivation for the ordinary product of functions and $Z=0$ vanishes identically. 
The cosymplectic case is a particular case of Poisson manifolds.
 If we drop the integrability conditions, we say we have an {\it almost Jacobi manifold.} Or almost Poisson if it is the case.
\begin{definition}
 Given two Jacobi manifolds $(M_1,\Lambda_1,Z_1)$ and $(M_2,\Lambda_2,Z_2)$ we say that the map $\phi:{M_1}\rightarrow M_2$ is a {\it Jacobi map}
if given two functions $f,g$ that are $C^{\infty}$ on $M_2$,
\begin{equation*}
 \{f\circ \phi,g\circ \phi\}_{M_1}=\{f,g\}_{M_2}\circ \phi
\end{equation*}

\end{definition}

\begin{theorem}
The characteristic distribution of a Jacobi manifold $(M,\Lambda,Z)$ is completely integrable in the sense of Stefan--Sussmann \cite{Stefan, Sussmann}, thus
$M$ defines a foliation whose leaves are not necessarily of the same dimension, and it is called the {\it characteristic foliation}. Each leaf has a unique
transitive Jacobi structure such that its canonical injection into $M$ is a Jacobi map.
Each leaf defines
\begin{enumerate}
 \item A locally conformally symplectic manifold if the dimension is even.
\item A manifold equipped with a contact one-form if its dimension is odd.
\end{enumerate}

\end{theorem}

In the case of {\it locally conformally symplectic structures}, we have a manifold $M$ that is even dimensional and a pair $(\Omega, \eta)$
where $\Omega$ is a 2-form and $\eta$ is a one form, such that $\Omega$ is everywhere of rank $2n=\text{dim} M$ and satisfy
\begin{equation}
 d\eta=0,\quad d\Omega+\eta\wedge \Omega=0
\end{equation}
Let $Z$ be the vector field and $\Lambda$ be unique such that
\begin{equation}
 \iota_Z\Omega=\eta,\quad \iota_{\sharp(Z)}\Lambda=Z,
\end{equation}
where $\sharp$ is the morphism induced by $\Lambda:T^{*}M\rightarrow TM$.
Then, $(M,\Lambda,Z)$ is a Jacobi manifold. It can be easily verified by observing that on the neighborhood of each point, there exists a function
$f$ such that $\eta=df$ and the locally defined two-form $e^{f}\Omega$ is symplectic.

In the case of symplectic manifolds, we have a pair $(M,\Omega)$, where $\Omega$ is a symplectic two-form. We define the map  
\begin{equation}\label{inv}
\flat:TM\rightarrow T^{*}M \quad \text{such that}\quad  \flat(X)=\iota_X\Omega
\end{equation}
is an isomorphism and nondegenerate. We can define its inverse as $\sharp:T^{*}M\rightarrow TM$ such that
\begin{equation*}
 \{f,g\}=\Omega(\sharp(df),\sharp(dg))=\langle dg,\sharp(df)\rangle=-\langle df, \sharp(dg)\rangle
\end{equation*}
satisfies the Jacobi identity, and it is a class of Jacobi manifold.
In particular, a symplectic manifold $(M,\Omega)$ is a Poisson manifold. Its bivector $\Lambda$ is such that the associated vector bundle
map $\sharp:T^{*}M\rightarrow TM$ defined by $\langle df,\sharp (dg)\rangle=\Lambda(df,dg)$ is the inverse of \eqref{inv}.

From there, we shall depict two particular cases that we shall use along the paper.

An {\it almost cosymplectic structure} on a $2n+1$-dimensional manifold $M$ consists of the triple $(M,\eta,\Omega)$, where $\eta$ is a one-form and $\Omega$ is a two-form such that $\eta\wedge \Omega^n\neq 0$.
An almost cosymplectic manifold is equipped with the isomorphism of $C^{\infty}$-modules $\flat:\mathfrak{X}(M) \rightarrow \Lambda^{1}(T^{*}M)$
such that
\begin{equation}\label{betamorph}
 \flat(X)=i_Xd\eta+\eta(X)\eta
\end{equation}
where $X\in \mathfrak{X}(M)$, $\eta\in \Lambda^{1}(T^{*}M).$
\begin{theorem}
 If $(M,\eta,\Omega)$ is an almost cosymplectic structure, then there exists a unique vector field $\mathcal{R}$, the so-called Reeb vector field such that
\begin{equation}
 \iota_{\mathcal{R}}\eta=1,\quad \iota_{\mathcal{R}}\Omega=0.
\end{equation}
\end{theorem}

\begin{definition}
We say that an almost cosymplectic structure $(M,\Omega,\eta)$ is a {\it cosymplectic structure} if $d\eta=0$ and $d\Omega=0$.
A cosymplectic manifold is equipped with the $\flat$ isomorphism in \eqref{betamorph} and the Reeb vector field is retrieved as $\mathcal{R}=\flat^{-1}(\eta).$
\end{definition}

\begin{definition}
 We say that a pair $(M,\eta)$ is a {\it contact structure} if is an almost cosymplectic structure and $\Omega=d\eta.$
\end{definition}
From here, we refer to $\eta$ as a contact form and $(M,\eta)$ a contact manifold.
Given a contact manifold $M,\Omega=d\eta$ and the Reeb vector field $\mathcal{R}$. Let $\sharp:T^{*}M\rightarrow TM$ be the vector bundle
map such that for each $\alpha$ in $T^{*}M$, we have
\begin{equation}
 \iota_{\sharp(\alpha)}\eta=0,\quad \iota_{\sharp(\alpha)}d\eta=-(\alpha-<\alpha,\mathcal{R}>\eta)
\end{equation}
for all one-form $\alpha$ on $M$.
We define the 2-tensor $\Lambda$ by
\begin{equation}
 \Lambda(\alpha,\beta)=<\beta,\sharp(\alpha)>=-<\alpha,\sharp(\beta)>,\quad \alpha,\beta\in T^{*}M
\end{equation}
Then the triple $(M,\Lambda,\mathcal{R})$ defines a Jacobi structure determined by the contact form $\eta$.

Let us now consider the contact structure $(M,\eta)$. We have the following theorem.

\begin{theorem}
If $(E,\eta)$ is a contact structure, then there exists a unique vector field $\mathcal{R}$, the so-called Reeb vector field
such that
\begin{equation}
 i_{\mathcal{R}}d\eta=0,\quad i_{\mathcal{R}}\eta=1.
\end{equation}
\end{theorem}

\begin{proof}
We can choose a local coordinate system $(t,x_1,\dots,x_{n},y_1,\dots,y_n)$ on $E$, a point $p\in E$ and an open subset $U\subset E$ in
the neighborhood $p$, we can express the contact one-form as
\begin{equation*}
 \eta|_{U}=dt+x_1dx_2+\dots+y_{n-1}dy_{n}.
\end{equation*}
The Reeb vector field is $\mathcal{R}=\frac{\partial}{\partial t}$, because
\begin{equation}
  \left(\eta|_{U}\right)\left(\frac{\partial}{\partial t}\right)=1,\qquad \iota_{\frac{\partial}{\partial t}}\Omega=0
\end{equation}
Given that the associated system with $\Omega|_U$ is $\frac{\partial}{\partial t}$, this demonstrates the existence of a unique Reeb vector field $\mathcal{R}=\frac{\partial}{\partial t}.$

\end{proof}

\section{Lagrangian--Legendrian submanifolds}

Let $(M,\Lambda,Z)$ be a Jacobi manifold with characteristic distribution $\mathcal{C}$.

\begin{definition}
 A submanifold $N$ of a Jacobi manifold $(M,\Lambda,Z)$ is said to be a {\it lagrangian-legendrian submanifold} if the following equality holds
\begin{equation}
 \sharp (TN^{\circ})=TN \cap \mathcal{C},
\end{equation}
where $TN^{\circ}$ denotes the annihilator of $TN$.

\end{definition}

Assume that $(M,\Lambda)$ is a transitive Poisson manifold, that implies $\mathcal{C}=TM$. Then, we have that a submanifold
$N$ of $M$ is a lagrangian submanifold if and only if
\begin{equation}
 \sharp(TN^{\circ})=TN
\end{equation}

For a tangent vector bundle $(TM,\Omega)$, we consider a submanifold $N\subset M$. We define the {\it $\Omega$-orthogonal complementary} of $TN$ as
\begin{equation}
TN^{\bot}=\{\xi\in TN|\quad \Omega(\xi,\chi)=0,\quad \forall \chi \in TN\}
\end{equation}
We say that a submanifold $N$ is {\it isotropic} if $TN\subset TN^{\bot}$, that is $\Omega(\xi,\chi)=0, \forall \xi,\chi\in N$. We say that it is {\it lagrangian}
if $TN$ is isotropic and has an isotropic complementary. That is, $TM=TN\oplus TN^{\bot}$. Now, the following assertions are equivalent,
\begin{enumerate}
 \item $N$ is lagrangian
\item $TN=TN^{\bot}$
\item $TN$ is isotropic and $\text{rank}(TN)=\frac{1}{2}\text{rank}(TM)$ (with rank we refer to the rank of the vector bundle).
\end{enumerate}
As a consequence, we characterize a lagrangian submanifold by checking if it has half dimension of $M$ and $\Omega|_{TN}=0$.

\begin{definition}
 We say that a submanifold $N\subset M$ of a cosymplectic manifold $(M,\eta,\Omega)$ is {\it lagrangian} if
\begin{equation}
 \sharp (TN^{\circ})=TN
\end{equation}
\end{definition}

\begin{definition}
We say that a submanifold $N$ of a contact manifold $(M,\eta)$ is {\it legendrian} if the following condition is fulfilled
\begin{equation}
 \sharp (TN^{\circ})=TN
\end{equation}
\end{definition}

\begin{proposition}
A submanifold $N$ of a contact manifold $(M,\eta)$ is a {\it legendrian} submanifold if and only if it is an integral manifold of maximal dimension $n$ of the
distribution $\eta=0$.
\end{proposition} 

\begin{proof}
 Assume that $M$ has dimension $2n+1$. If a submanifold $N$ of $M$ is legendrian then the condition
\begin{equation*}
 \sharp(TN^{\circ})=TN
\end{equation*}
implies that $\eta|_{N}=0$. Moreover, $N$ has necessary dimension $n$, since $T_xN$ will be a lagrangian subspace
of the symplectic vector space $(\ker{\eta}_x,(d\eta)_x)$ for all $x\in N$. The converse is proved reversing the arguments.

\end{proof}

\section{Hamilton--Jacobi theory on cosymplectic manifolds}

\subsection{Geometric approach}
Consider the extended phase space $T^{*}Q\times \mathbb{R}$ and its canonical projections of the first and second factor, $\rho:T^{*}Q\times \mathbb{R}\rightarrow T^{*}Q$ and  $t:T^{*}Q\times \mathbb{R}\rightarrow \mathbb{R}$, respectively
 and a time-dependent Hamiltonian $H:T^{*}Q\times \mathbb{R}\rightarrow \mathbb{R}$. It can be described through the following diagram

\[
\xymatrix{ T^{*}Q\times \mathbb{R}
\ar[dd]^{\rho}  \ar[ddrr]^{t}   &   &\\
  &  & &\\
T^{*}Q & & \mathbb{R}}
\]

\bigskip
\noindent
We have canonical coordinates $\{q^i,p_i,t\}$ with $i=1,\dots,n$.
We consider the two-form on $T^{*}Q\times \mathbb{R}$ as $\Omega_H=-d\theta_H$
and $\theta_H=\theta_Q-Hdt$ where $\theta_Q$ is the canonical Liouville one-form. Hence, $\Omega_H=dq^i\wedge dp_i+dH\wedge dt$.
Then, we have a cosymplectic structure $(dt,\Omega_H).$
The corresponding Reeb vector field needs to satisfy $\iota_{\mathcal{R}}dt=1,\iota_{\mathcal{R}}\Omega_H=0$, then it reads
\begin{equation}\label{reebcosym}
 \mathcal{R}_H=\frac{\partial}{\partial t}+\sum_{i=1}^n\frac{\partial H}{\partial p_i}\frac{\partial}{\partial q^i}-\sum_{i=1}^n \frac{\partial H}{\partial q^i}\frac{\partial}{\partial p_i}.
\end{equation}

The corresponding classical Hamilton--Jacobi equations are
\begin{equation}\label{hamileq22}
\left\{\begin{aligned}
 {\dot q}^i&=\frac{\partial H}{\partial p_i},\\
 {\dot p}_i&=-\frac{\partial H}{\partial q^i}\\
{\dot t}&=1.
 \end{aligned}\right.
 \end{equation}

for all $i=1,\dots,n$. 
\noindent
We consider the fibration
$\pi: T^{*}Q\times \mathbb{R}\rightarrow Q\times \mathbb{R}$ and a section $\gamma$ of $\pi:T^{*}Q\times \mathbb{R} \rightarrow Q\times \mathbb{R}$, i.e., $\pi\circ \gamma=\text{id}_{Q\times \mathbb{R}}$. 
Also, we assume that $\gamma_t:Q\rightarrow T^{*}Q$ is a lagrangian submanifold of the cosymplectic manifold $(T^{*}Q\times \mathbb{R},dt,\Omega_H)$ for a fixed time,
\[
 \xymatrix{T^{*}Q\times \mathbb{R}\ar[d]^{\rho}& & &\\
 T^{*}Q \ar[d]^{\pi}& & &\\
Q\times \mathbb{R}\ar[d]& & &  \\
Q\ar@/^2pc/[uuu]^{\gamma_t}}
\]
that is $d\gamma_t=0$.

We can use $\gamma$ to project $\mathcal{R}_H$ on $Q\times \mathbb{R}$
just defining a vector field $\mathcal{R}^{\gamma}_H$ on $Q\times \mathbb{R}$ by
\begin{equation}\label{hjr}
 \mathcal{R}^{\gamma}_H=T_{\pi}\circ \mathcal{R}_H\circ \gamma
\end{equation}
The following diagram summarizes the above construction
\[
\xymatrix{ T^{*}Q\times \mathbb{R}
\ar[dd]^{\pi} \ar[rrr]^{\mathcal{R}_H}&   & &T(T^{*}Q\times \mathbb{R})\ar[dd]^{T_{\pi}}\\
  &  & &\\
Q\times \mathbb{R} \ar@/^2pc/[uu]^{\gamma}\ar[rrr]^{\mathcal{R}^{\gamma}_H}&  & & T(Q\times \mathbb{R})}
\]
\begin{definition}
 If $\alpha$ is a one-form, locally expressed as $\alpha=\alpha_i dq^i$, we designate by $\alpha^{V}$ the {\it vertical lift}
or vector fields associated with $\alpha$, defined by
\begin{equation*}
 \iota_{\alpha^V}\omega_Q=\alpha
\end{equation*}
Hence, the vector field takes the form 
\begin{equation}
\alpha^{V}=-\alpha_i\frac{\partial}{\partial p_i}
\end{equation}
\end{definition}

\begin{theorem}
 The vector fields $\mathcal{R}_H$ and $\mathcal{R}^{\gamma}_H$ are $\gamma$-related if and only if the following equation is satisfied
\begin{equation}\label{eqtheorem2}
[d(H\circ \gamma_t)]^{V}=\dot{\gamma}_q
\end{equation}
where $[\dots]^{V}$ denotes the vertical lift of a one-form on $Q$ to $T^{*}Q$.  Now $\dot{\gamma}_q$ is the tangent vector in a point $q$ associated with the curve
\[
\xymatrix{
 \mathbb{R}\ar[rr]\ar@/^2pc/[rrrrrr]^{\gamma_q} & &  Q\times \mathbb{R}\ar[rr] & &  T^{*}Q\times \mathbb{R}\ar[rr]^{\rho} & &  T^{*}Q}
\]

\end{theorem}
\begin{proof}
 The vector fields $\mathcal{R}_H$ and $\mathcal{R}^{\gamma}_H$ are $\gamma$ related if $T\gamma(\mathcal{R}^{\gamma}_H)=\mathcal{R}_H$. That is,
\begin{equation}\label{transf11}
T\gamma(\mathcal{R}^{\gamma}_H)=T\gamma\left(\frac{\partial}{\partial t}+\sum_{i=1}^n\frac{\partial H}{\partial p_i}\frac{\partial}{\partial q_i}\right)
\end{equation}
We choose a section $\gamma(q^i,t)=\gamma(q^i,\gamma^j(q^i,t),t)$ with $i,j=1,\dots,n$ such that the lift in the tangent bundle reads,
\begin{equation}\label{tangent2}
T\gamma\left(\frac{\partial}{\partial t}\right)=\frac{\partial}{\partial t}+\sum_{j=1}^n\frac{\partial \gamma^j}{\partial t}\frac{\partial}{\partial p_j},\quad T\gamma\left(\frac{\partial}{\partial q^i}\right)=\frac{\partial}{\partial q^i}+\sum_{j=1}^n\frac{\partial \gamma^j}{\partial q^i}\frac{\partial}{\partial p_j} 
\end{equation}
Introducing equations \eqref{tangent2} in equation \eqref{transf11}, it is straightforward to retrieve condition \eqref{eqtheorem2} if we use that 
$\gamma_t$ is closed for permuting indices in intermediate steps to obtain \eqref{eqtheorem2}.
\end{proof}
Equation \eqref{eqtheorem2} is known as a {\it Hamilton--Jacobi equation on a cosymplectic manifold.}
In local coordinates, 
\begin{equation}\label{lc}
 \frac{\partial \gamma^j}{\partial t}+\sum_{i=1}^n \frac{\partial H}{\partial p_i}\frac{\partial \gamma^j}{\partial q^i}+\frac{\partial H}{\partial q^j}=0.
\end{equation}

\subsection{Applications}

\subsubsection*{Winternitz--Smorodinsky oscillator}
We consider a Hamiltonian formalism for the one-dimensional {\it Winternitz--Smorodinsky oscillator} \cite{indio}. The following superintegrable Hamiltonian
retrives the Hamilton equations of a nonlinear oscillator which has been contemplated in \cite{sardontesis} and has multiple applications in both Classical and Quantum Mechanics \cite{ws}.
When $k=0$, we recover the well-known isotropic harmonic oscillator. It reads
\begin{equation}
 H=\frac{1}{2}\left(p^2+\frac{k}{q^2}\right)+\frac{1}{2}\omega(t)^2q^2
\end{equation}
The defined Reeb vector field in \eqref{reebcosym} applied in this case, reads
\begin{equation}
 \mathcal{R}=\frac{\partial}{\partial t}+p\frac{\partial}{\partial q}-\left(\omega(t)^2q-\frac{k}{q^3}\right)\frac{\partial}{\partial p},
\end{equation}

We consider a section $\gamma(q,t)=(q,\gamma(q,t),t)$ (notice that we have abused language by denoting the triple $\gamma(q,t)=(q,\gamma(q,t),t)$ and its middle entry by $\gamma(q,t)$) such that the Hamilton--Jacobi equation reads
\begin{equation}\label{hjmp}
 \frac{\partial \gamma}{\partial t}+\gamma \frac{\partial \gamma}{\partial q}=\frac{k}{q^3}-\omega^2(t)q.
\end{equation}
and $\mathcal{R}^{\gamma}$ is
\begin{equation}
 \mathcal{R}^{\gamma}=\frac{\partial}{\partial t}+p\frac{\partial}{\partial q}.
\end{equation}
This equation can be solved by the method of characteristics \cite{Stephani}
\begin{equation}\label{char}
 dt=\frac{dq}{\gamma}=\frac{d\gamma}{\frac{k}{q^3}-\omega^2(t)q}.
\end{equation}
Integrating the equations along the section $\gamma$, we have that $p=\gamma$, then, we can solve system \eqref{char} whose
solutions result in,
\begin{align*}
 \frac{dq}{dt}&=\gamma,\\
\frac{d\gamma}{dt}&=\frac{k}{q^3}-\omega^2(t)q,
\end{align*}
which is a {\it Milne--Pinney equation} \cite{Milne1,Milne2} that has received a lot of expectation due to its ubiquity in Physics and Engineering. It is a model for propagation
of laser beams in nonlinear media and also applies to plasma dynamics. It has the expression
\begin{equation}
 \frac{d^2q}{dt^2}=\frac{k}{q^3}-\omega^2(t)q.
\end{equation}
Its solution can be expressed in terms of two solutions $y_1,y_2$ of the harmonic oscillator and two constants.
\begin{equation}
 q=\frac{\sqrt{2}}{|W|}\sqrt{C_1y_1^2+C_2y_2^2\pm \sqrt{4C_1C_2-kW^2y_1y_2}}
\end{equation}
where $W$ is the Wronskian $W=y_1\dot{y}_2-y_2\dot{y}_1$ of the two solutions $y_1$ and $y_2$ of the associated harmonic oscillator.
Then, $\gamma$ takes the form
\begin{equation}
 \gamma=\frac{d}{dt}\left(\frac{\sqrt{2}}{|W|}\sqrt{C_1y_1^2+C_2y_2^2\pm \sqrt{4C_1C_2-kW^2y_1y_2}}\right)
\end{equation}

\subsubsection*{A trigonometric system}
Let us consider the time-dependent Hamiltonian
\begin{equation}
 H=\frac{p^2}{2}+\frac{q^2}{2}+\alpha \sin{(wt)}\frac{q^2p^2}{2}.
\end{equation}
In our setting, the two-form $\Omega_H=dq\wedge dp+dH\wedge dt.$
The Reeb vector field reads
\begin{equation}
 \mathcal{R}_H=\frac{\partial}{\partial t}+\left(p+\alpha\sin{(wt)}q^2p\right)\frac{\partial}{\partial q}-\left(q+\alpha\sin{(wt)}p^2q\right)\frac{\partial}{\partial p}.
\end{equation}
We choose a lagrangian section $\gamma=(q,\gamma(q,t),t)$.
The $\mathcal{R}_H^{\gamma}$ field is
\begin{equation}
 \mathcal{R}^{\gamma}_H=\frac{\partial}{\partial t}+\left(p+\alpha\sin{(wt)}q^2p\right)\frac{\partial}{\partial q}. 
\end{equation}
If we impose \eqref{hjr} to be fulfilled, we need to compute the terms
\begin{equation}
 T\gamma\left(\frac{\partial}{\partial t}\right)=\frac{\partial}{\partial t}+\frac{\partial \gamma}{\partial t}\frac{\partial}{\partial p},\quad T\gamma\left(\frac{\partial}{\partial q}\right)=\frac{\partial}{\partial q}+\frac{\partial \gamma}{\partial q}\frac{\partial}{\partial p},
\end{equation}
The arising equation reads
\begin{equation}
 \frac{\partial \gamma}{\partial t}+\left(p+\alpha\sin{(wt)}q^2p\right)\frac{\partial \gamma}{\partial q}=q+\alpha\sin{(wt)}p^2q.
\end{equation}
This equation is a quasi-linear first-order PDE for a function $\gamma(q,t)$. It can be solved with the aid of the method of characteristics \cite{evans}
\begin{equation}
 dt=\frac{dq}{p+\alpha \sin{(wt)}q^2p}=\frac{d\gamma}{q+\alpha \sin{(wt)}p^2q}
\end{equation}
which turns in the following system of equations
\begin{align}\label{sysgp}
 \frac{dq}{dt}&=p(1+\alpha \sin{(wt)}q^2),\nonumber\\
\frac{d\gamma}{dt}&=q(1+\alpha \sin{(wt)}p^2).
\end{align}
Integrating the equations along the section $\gamma$, we have that $p=\gamma$, then, we can solve system \eqref{sysgp} whose
solutions result in
 \begin{equation}
  \gamma=\pm\frac{e^{2t}+2C_2}{\sqrt{-\alpha\sin{(wt)}e^{4t}+4\alpha \sin{(wt)}e^{2t}C_2-4\alpha \sin{(wt)}C_2^2+4e^{2t}C_1}}
 \end{equation}
where $C_1$ and $C_2$ are arbitrary constants of integration.

\section{Hamilton--Jacobi theory on contact manifolds}

We consider the extended phase space $T^{*}Q\times \mathbb{R}$  with canonical projections of the first and second variables
$\rho:T^{*}Q\times \mathbb{R}\rightarrow T^{*}Q$ and $t:T^{*}Q\times \mathbb{R}\rightarrow \mathbb{R}$. The Hamiltonian function is
$H:T^{*}Q\times \mathbb{R}\rightarrow \mathbb{R}$. It can be illustrated through the following diagram

\[
\xymatrix{ T^{*}Q\times \mathbb{R}
\ar[dd]^{\rho} \ar[ddrr]^{t}\\
  &  & &\\
T^{*}Q &  & \mathbb{R}}
\]

\bigskip
\noindent
We have local canonical coordinates $\{q^i,p_i,t\}, i=1,\dots,n$. 
The one form is $\eta=dt-\rho^{*}\theta_Q$.
The pair $(T^{*}Q\times \mathbb{R},\eta)$ is a contact manifold.
The Reeb vector field is $\mathcal{R}=\frac{\partial}{\partial t}$ that fulfills 
$$\iota_{\mathcal{R}}\eta=1,\quad \iota_{\mathcal{R}}d\eta=0.$$

Consider the fibration $\pi: T^{*}Q\times \mathbb{R}\rightarrow Q\times \mathbb{R}$.
To have dynamics, we consider a vector field \cite{mejicanos} defined by
\begin{equation*}
 \flat{(X_H)}=-(\mathcal{R}(H)+H)\eta+dH.
\end{equation*}
where $\flat$ is the isomorphism defined in \eqref{betamorph}. 
In particular, the vector field $X_H$ fulfills the following condition
\begin{equation}\label{1exph}
 \eta(X_H)=-H.
\end{equation}
In coordinates, $\eta=dt-\sum_{i=1}^n p_idq^i$ and $X_H$ satisfying condition \eqref{1exph} takes the form
{\begin{footnotesize}
\begin{equation}\label{1hvf}
 X_H=\sum_{i=1}^n\left(p_i\frac{\partial H}{\partial p_i}-\sum_{i=1}^nH\right)\frac{\partial}{\partial t}-\sum_{i=1}^n\left(p_i\frac{\partial H}{\partial t}+\frac{\partial H}{\partial q^i}\right)\frac{\partial}{\partial p_i}+\sum_{i=1}^n\frac{\partial H}{\partial p_i}\frac{\partial}{\partial q^i}
\end{equation}
\end{footnotesize}}

The reason for choosing $X_H$ as in \eqref{1hvf}, proposed in \cite{mejicanos}
resides in the retrieval of the {\it dissipation Hamilton--Jacobi equations}. Such equations read

\begin{equation}\label{hamileq}
\left\{\begin{aligned}
 {\dot q}^i&=\frac{\partial H}{\partial p_i},\\
 {\dot p}_i&=-\frac{\partial H}{\partial q^i}-p_i\frac{\partial H}{\partial t},\\
{\dot t}&=p_i\frac{\partial H}{\partial p_i}-H.
 \end{aligned}\right.
 \end{equation}
for all $i=1,\dots,n$. 

Consider $\gamma$ a section of $\pi:T^{*}Q\times \mathbb{R} \rightarrow Q\times \mathbb{R}$, i.e., $\pi\circ \gamma=\text{id}_{Q\times \mathbb{R}}$. We can use $\gamma$ to project $X_H$ on $Q\times \mathbb{R}$
just defining a vector field $X_{H}^{\gamma}$ on $Q\times \mathbb{R}$ by
\begin{equation}\label{hjpar}
 X_H^{\gamma}=T_{\pi}\circ X_{H}\circ \gamma
\end{equation}
The following diagram summarizes the above construction
\[
\xymatrix{ T^{*}Q\times \mathbb{R}
\ar[dd]^{\pi} \ar[rrr]^{X_H}&   & &T(T^{*}Q\times \mathbb{R})\ar[dd]^{T_{\pi}}\\
  &  & &\\
Q\times \mathbb{R} \ar@/^2pc/[uu]^{\gamma}\ar[rrr]^{X^{\gamma}_H}&  & & T(Q\times \mathbb{R})}
\]

Assume that $\gamma(Q\times \mathbb{R})$ is a legendrian submanifold, such that $\gamma_t$ is closed.

\begin{theorem}
 The vector fields $X_H$ and $X_H^{\gamma}$ are $\gamma$-related if and only if the following equation is satisfied
\begin{equation}\label{eqtheorem}
[d(H\circ \gamma)]^{V}=-H\dot{\gamma}_q
\end{equation}
where $[\dots]^{V}$ denotes the vertical lift of a one-form and $\dot{\gamma}_q$ is the tangent vector in a point $q$ associated with the curve
\[
\xymatrix{
 \mathbb{R}\ar[rr]\ar@/^2pc/[rrrrrr]^{\gamma_q} & &  Q\times \mathbb{R}\ar[rr] & &  T^{*}Q\times \mathbb{R}\ar[rr]^{\rho} & &  T^{*}Q}
\]

\end{theorem}

\begin{proof}
 The vector fields $X_H$ and $X_H^{\gamma}$ are $\gamma$ related if $T\gamma (X_H^{\gamma})=X_H$. That is,
\begin{equation}\label{transf1}
T\gamma (X_H^{\gamma})=\left(p_i\frac{\partial H}{\partial p_i}-H\right)T\gamma\left(\frac{\partial}{\partial t}\right)+\frac{\partial H}{\partial p_i}T\gamma\left(\frac{\partial}{\partial q^i}\right)=X_H
\end{equation}
The section in local coordinates is $\gamma(q^i,t)=\gamma(q^i,\gamma^j(q^i,t),t)$ with $i,j=1,\dots,n$ such that the lift in the tangent bundle reads,
\begin{equation}\label{tangent3}
T\gamma\left(\frac{\partial}{\partial t}\right)=\frac{\partial}{\partial t}+\sum_{j=1}^n\frac{\partial \gamma^j}{\partial t}\frac{\partial}{\partial p_j},\quad T\gamma\left(\frac{\partial}{\partial q^i}\right)=\frac{\partial}{\partial q^i}+\sum_{j=1}^n \frac{\partial \gamma^j}{\partial q^i}\frac{\partial}{\partial p_j},\\
\end{equation}
Introducing equations \eqref{tangent3} in equation \eqref{transf1}, it is straightforward to retrieve condition \eqref{eqtheorem} if 
a further condition on the one-form $\gamma$ is imposed. It is
\begin{equation}
 d\gamma_t=0
\end{equation}
That is, $\gamma_t$ is closed and fulfils the legendrian submanifold condition.
\end{proof}
Equation \eqref{eqtheorem} is known as a {\it Hamilton--Jacobi equation with respect to a contact structure.}
In local coordinates, 
\begin{equation}
 p_j\frac{\partial H}{\partial t}+\frac{\partial H}{\partial q^j}+\sum_{i=1}^n \left(p_i\frac{\partial H}{\partial p_i}-H\right)\frac{\partial \gamma^j}{\partial t}+\sum_{i=1}^n \frac{\partial H}{\partial p_i}\frac{\partial \gamma^j}{\partial q^i}=0
\end{equation}

%
%

\subsection{Applications}
Let us consider the Hamiltonian 
\begin{equation}
 H=\frac{p^2}{2m}+V(q)+\alpha S
\end{equation}
This is the corresponding Hamiltonian of a {\it damped oscillator} \cite{mejicanos} which is retrieved by \eqref{hamileq}.
Taking the Hamiltonian vector field as in \eqref{1hvf}, we have
\begin{equation}
 X_H=\left(\frac{p^2}{2m}-V(q)-\alpha S\right)\frac{\partial}{\partial S}-\left(\alpha p+ V'(q)\right)\frac{\partial}{\partial p}+\frac{p}{m}\frac{\partial}{\partial q}
\end{equation}
We choose a legendrian section $\gamma=(q,\gamma(q,S),S)$. And $X_H^{\gamma}$ reads
\begin{equation}
 X_H^{\gamma}=\left(\frac{p^2}{2m}-V(q)-\alpha S\right)\frac{\partial}{\partial S}+\frac{p}{m}\frac{\partial}{\partial q}
\end{equation}
Using \eqref{hjpar}, we need to perfom the computations
\begin{equation}
 T\gamma\left(\frac{\partial}{\partial S}\right)=\frac{\partial}{\partial S}+\frac{\partial \gamma}{\partial S}\frac{\partial}{\partial p},\quad T\gamma\left(\frac{\partial}{\partial q}\right)=\frac{\partial}{\partial q}+\frac{\partial \gamma}{\partial q}\frac{\partial}{\partial p}
\end{equation}
The Hamilton Jacobi equation reads
\begin{equation}\label{puf1}
 \left(\frac{p^2}{2m}-V(q)-\alpha S\right)\frac{\partial \gamma}{\partial S}+\frac{p}{m}\frac{\partial \gamma}{\partial q}+(p\alpha+V'(q))=0
\end{equation}
with $d\gamma_S=0$, that is $\gamma_S=\text{constant}$. Integrating the equations along the section $\gamma$, we have that $p=\gamma$, 
and setting the constant $\gamma_S=1$, then \eqref{puf1} can be rewritten as
\begin{equation}\label{puf2}
 \frac{\partial \gamma}{\partial q}+\frac{1}{2}\gamma+\alpha m+\frac{m}{\gamma}\left(V'(q)-V(q)-\alpha S\right)=0
\end{equation}
which can be solved as
\begin{equation}
 q=\frac{c_1}{\sqrt{c_1^2-2c_2}}\ln{\left(\frac{\gamma+c_1-\sqrt{c_1^2-2c_2}}{\gamma+c_1+\sqrt{c_1^2-2c_2}}\right)}-\ln{\left(\frac{1}{2}\gamma^2+c_1\gamma+c_2\right)}
\end{equation}
when $c_1^2>2c_2$, and
\begin{equation}
 q=\frac{2}{\sqrt{2c_2-c_1^2}}\tan^{-1}{\left(\frac{\gamma+c_1}{\sqrt{2c_2-c_1^2}}\right)}-\ln{\left(\frac{1}{2}\gamma^2+c_1\gamma+c_2\right)}
\end{equation}
when $c_2>\frac{c_1^2}{2}$, with
\begin{equation}
 c_1=\alpha m,\quad c_2=-m\left(V(q)-V'(q)+m\alpha S\right).
\end{equation}

\section{Conclusions}
We have developed a two-fold geometric Hamilton--Jacobi theory: for time dependent Hamiltonians through a cosymplectic geometric formalism and for dissipative
Hamiltonians through a contact geometry formalism.
 We have derived an expression
for the Hamilton--Jacobi equation and have applied our result to time dependent, unidimensional Winternitz--Smorodinsky Hamiltonian. The Hamilton Jacobi equation for the lagrangian section
$\gamma$ is a Milne--Pinney equation which can be integrated with the aid of the Lie systems theory.
Furthermore, we have developed a geometric Hamilton--Jacobi theory for Hamiltonians containing a dissipation term. We have derived an expression for the Hamilton--Jacobi
equation and have applied our result to a one dissipation parameter dependent Hamiltonian. This latter case has been developed along the lines of contact geometry.

\section*{Acknowledgements}
This work has been partially supported by MINECO MTM 2013-42-870-P and
the ICMAT Severo Ochoa project SEV-2011-0087.


\begin{thebibliography}{10}

\bibitem{AbraMars}
R. Abraham and J.E. Marsden,
{\sl Foundations of Mechanics,} 2nd Ed.
Benjamin--Cumming, Reading, 1978.

\bibitem{Albert}
C. Albert,
J. Geom. Phys. {\bf 6}, 627--249 (1989).



\bibitem{Arnold}
V.I. Arnold,
{\sl Mathematica methods of Classical Mechanics,} Graduate Texts in Mathematics {\bf 60},
Springer--Verlag, Berlin, 1978.

\bibitem{indio}
J.N. Bandyopadhyay and V.B. Sheorey,
Phys. Rev. A {\bf 63}, 042109 (2001).

\bibitem{Becker}
K. Becker, M. Becker and A. Strominger,
 Nuclear Phys. B {\bf 456}, 130--152 (1995).

\bibitem{Blair}
D.E. Blair,
{\sl Contact manifolds in Riemannian Geometry},
Lecture Notes in Mathematics,
Springer, 1976.




 \bibitem{boothwang}
W.H. Boothby and H.C. Wang,
Annals of Mathematics, {\bf 68}, 721--734 (1958).


\bibitem{mejicanos}
A. Bravetti, H. Cruz and D. Tapias,
https://arxiv.org/abs/1604.08266. (2016).

\bibitem{CLMMdD}
F. Cantrijn, M. de Le\'on, J.C. Marrero and D.M. de Diego,
Rep. Math. Phys. {\bf 42}, 25--45 (1998).

\bibitem{Cape}
B. Cappelleti-Montano, A. de Nicola and I. Yadin,
Rev. Math. Phys. {\bf 25}, 1343002 (2013).

\bibitem{CGMMMLRR}
J.F. Cari{\~n}ena {\it et. al}
Int. J. Geom. Meth. Mod. Phys. {\bf 7}, 431--454 (2010).

\bibitem{CGMMMLRR1}
J.F. Cari{\~n}ena {\it et al.}
Int. J. Geom. Meth. Mod. Phys. {\bf 3}, 1417--1458 (2006).

\bibitem{CGL11}
J.F. Cari\~nena, J. Grabowski and J. de Lucas,
J. Phys. A: Math. Theor. {\bf 45},  185202 (2012).

\bibitem{Milne1}
J.F. Cari\~nena and J. de Lucas,
Phys. Lett. A {\bf 372}, 5385--5389 (2008).

\bibitem{Milne2}
J.F. Cari\~nena and J. de Lucas,
Int. J. Geom. Methods Mod. Phys. {\bf 6}, 683 (2009).


\bibitem{CariNasarre}
J.F. Cari{\~n}ena and J. Nasarre,
Fortschritte der Physik, {\bf 44}, 181--198 (1996).


%
\bibitem{Cham}
 A.H. Chamseddine,
 Phys. Lett. B {\bf 233}, 291--294 (1989).


\bibitem{CLM}
D. Chinea, M. de Le\'on and J.C. Marrero,
J. Math. Phys. {\bf 35}, 3410--3447 (1994).

\bibitem{Crampin}
 M. Crampin,
 J. Phys. A: Math. Gen. {\bf 14}, 2567--2575 (1981).
%

\bibitem{echevarria}
A. Echevarr\'ia Enriquez, M.C. Mu{\~n}oz Lecanda and N. Rom\'an Roy,
{\it Rev. Math. Phys.} {\bf 3}, 301--330 (1991).


 \bibitem{etnyre}
 J.B. Etnyre,
{\sl Contact manifolds},\\
 people.math.gatech.edu/~etnyre/preprints/papers/contlect.pdf.

\bibitem{evans}
L.C. Evans,
{\sl Partial differential equations.} Providence: American Mathematical Society, 1998.



\bibitem{Gold}
 H. Goldstein, 
 {\sl Mec\'anica Cl\'asica,} 4a Ed. Aguilar SA 
 Madrid 1979.

 \bibitem{Godbillon}
 G. Godbillon,
{\sl Geometrie differentielle et mecanique analytique},
 Collection Methodes, Hermann, Paris, 1989.

\bibitem{hamil1}
W. Hamilton,
Dublin University Review, 795--826 (1833).

\bibitem{hamil2}
W. Hamilton,
British Association Report, 513--518 (1834).


\bibitem{Hitchin}
 N. Hitchin,
 Q.J. Math. {\bf 54}, 281--308 (2003).



\bibitem{IboLeonMarmo}
A. Ibort, M. de Le\'on and G. Marmo,
J. Phys. A {\bf 30}, 2783--2798 (1997).







%
 \bibitem{Kibble}
 T.W. Kibble and F.H. Berkshire,
 {\sl Classical Mechanics}, Imperial College Press, London, 5th Ed., 2004.



 \bibitem{Klein}
 J. Klein,
 C.R. Acad. Sci. Paris, {\bf 257}, 2392--2394 (1963).


\bibitem{Laci}
I. Lacirasella, J.C. Marrero and E. Padr\'on,
J. Phys A {\bf 45}, 325202 (2012).


 \bibitem{LL}
 L.D. Landau and E.M. Lifshitz,
 {\sl Mec\'anica}, 2a Ed, v.{\bf I}, Academia de Ciencias URSS. Editoral Reverte, Barcelona, 1988.




\bibitem{LeonChineaMarrero}
M. de Le\'on, D. Chinea and J.C. Marrero,
J. Math. Pures Appl. {\bf 72}, 567--591 (1993).

\bibitem{LeonIglDiego}
M. de Le\'on, D. Iglesias-Ponte and D. Mart\'in de Diego,
{J. Phys. A: Math. Gen} {\bf 1}, 015205, 14 pp. (2008).

\bibitem{LeonMarinMarrero}
M. de Le\'on, J. Mar\'in Solano and J.C. Marrero,
{Differential Geom. Appl.} {\bf 6}, 275--300 (1996).


\bibitem{LeonDiegoMarrSalVil}
M. de Le\'on {\it et. al}
{Int. J. Geom. Meth. Mod. Phys.} {\bf 7}, 14911507 (2010).




\bibitem{LeonMarrero}
M. de Le\'on and J.C. Marrero,
J. Math. Phys. {\bf 34}, 622--644 (1993).

\bibitem{LeonMarrDiego08}
M. de Le\'on, J.C. Marrero and D. Mart\'in de Diego,
{\sl A geometric Hamilton--Jacobi theory for classical field theories}. In: {\it Variations, Geometry and Physics} 129--140,
Nova Sci. Publ., New York, 2009.

\bibitem{LeonMarrDiego}
M. de Le\'on, J.C. Marrero and D. Mart\'in de Diego,
{J. Geom. Mech.} {\bf 2}, 159--198 (2010).

\bibitem{LeonMarrDiegoVaq}
M. de Le\'on, J.C. Marrero, D. Mart\'in de Diego and M. Vaquero,
J. Phys. A {\bf 54}, 032902 32pp. (2013).

\bibitem{LMM}
M. de Le\'on, J.C. Marrero and E. Mart\'inez,
J. Math. Phys: Math. Gen. {\bf 38}, R241--R308 (2005).


\bibitem{LeonDiegoVaq2}
M. de Le\'on, D. Mart\'in de Diego and M. Vaquero,
{Int. J. Geom. Meth. Mod. Phys.} {\bf 9}, 125007 24pp. (2012).






\bibitem{LeonRodri}
M. de Leon and P.R. Rodrigues,
{\sl Methods of Differential Geometry in Analytical Mechanics,}
Mathematical Studies, North--Holland {\bf 158}, 1989.


\bibitem{LeonSara}
M. de Le\'on and M. Saralegi,
J. Phys. A {\bf 26}, 5033--5043 (1993).

\bibitem{LeonSardon}
M. de Le\'on and C. Sard\'on,
https://arxiv.org/abs/1604.08904. (2016).

\bibitem{LeonTuyn}
M. de Le\'on and G.M. Tuynman,
J. Geom. Phys. {\bf 20}, 77--86 (1996).


\bibitem{LiberMarle}
P. Libermann, C.M. Marle,
{\sl Symplectic geometry and analytical Mechanics} in: Mathematics and its applications {\bf 35},
D. Reidel Publishing Company, 1987.


\bibitem{Lich2}
A. Lichnerowicz,
C.R. Acad. Sci. Paris {\bf 253}, 1302--1304 (1961).

\bibitem{Lich}
A. Lichnerowicz,
J. Math. pures et Appl. {\bf 57}, 453--488 (1978).



\bibitem{Marle}
C.M. Marle,
{\sl On Jacobi manifolds and Jacobi bundles,} on:
Sympletic Geometry, Groupoids and Integrable systems,
Mathematics Sciences Research Institute Publications 227--246, 1991.

\bibitem{MarrSosa}
J.C. Marrero and D. Sosa,
{Int. J. Geom. Methods in Mod. Phys.} {\bf 3},
605--622 (2006).


\bibitem{MarsRatiu}
J.E. Marsden and T. Ratiu,
Lett. Math. Phys., {\bf 11}, 161--169 (1986).

\bibitem{MarsWein}
J.E. Marsden and A. Weinstein,
Rep. Math. Phys., {\bf 5}, 121--130 (1974).


\bibitem{Cedric}
C. Mart\'inez Campos, M. de Le\'on, D. Mart\'in de Diego and M. Vaquero,
Rep. Math. Phys. {\bf 76}, 359--387 (2015).

\bibitem{Nambu}
Y. Nambu,
{\it Phys. Rev. D}  {\bf 7}, 2405-2412 (1973).



\bibitem{oku}
M. Okumura,
K\"odai Math. Sem. Rep. {\bf 17}, 63--73 (1965).

%

 \bibitem{Otto}
 F. Otto,
 {\sl Lectures on Riemannian surfaces}, in:
 Graduate Texts in Mathematics {\bf 81},
 New York, Springer Verlag, 1981.

\bibitem{Pauffer}
 C. Pauffer and H. Romer,
 {\it J. Geom. Phys.} {\bf 44}, 52--69 (2002).

\bibitem{rajeev}
S.G. Rajeev,
Annals Phys. {\bf 323}, 768--782 (2008).



%

 \bibitem{Rund}
 H. Rund,
 {\sl The Hamilton--Jacobi theory in the calculus of variations},
 Robert E. Krieger, Publ. Co. Nuntington, NY, 1973.

\bibitem{sardontesis}
C. Sard\'on,
Lie systems, Lie symmetries and reciprocal transformations,\\
http://arxiv.org/abs/1508.00726. (2015).


\bibitem{Stefan}
 P. Stefan,
 Bull. Amer. Math. Soc. {\bf 80}, 1142--1145 (1974).

 \bibitem{Stein}
 K. Stein,
 Math. Ann. {\bf 123}, 201--222 (1951).


 \bibitem{Stephani}
H. Stephani,
{\sl Differential equations: their solution using symmetries,}
 Cambridge University Press, Cambridge, 1990.


 \bibitem{Sussmann}
 H.J. Sussmann, 
 Trans. Amer. Math. Soc. {\bf 180}, 171--180 (1973).

\bibitem{tulzcy1}
 W.M. Tulczyjew,
 C.R. Acad. Sci. Paris. Ser. A-B {\bf 283}, A15--A18, (1976).
%
 \bibitem{tulzcy2}
 W.M. Tulczyjew,
 C.R. Acad. Sci. Paris. Ser. A-B {\bf 283}, A675--A678, (1976).


\bibitem{We00}
 N. Weaver,
{J. Operator Theory}, {\bf 43}, 223--242 (2000).

\bibitem{ws}
P. Winternitz, A. Smorodinsky, M. Uhlir and J. Fris,
Soviet J. Nuclear Phys {\bf 4}, 444--450 (1967).



\end{thebibliography}
\end{document}